\documentclass{article}
\usepackage{hyperref}

\usepackage[utf8x]{inputenc}
\usepackage{ucs}

\usepackage{amsmath}
\usepackage{amsfonts}
\usepackage{amssymb}
\usepackage{stmaryrd} 

\IfFileExists{subcaption.sty}{
	\usepackage{subcaption}
	\captionsetup{compatibility=false}
}{
	\usepackage{caption}
	\DeclareCaptionSubType{figure}
	\makeatletter
	\caption@For{subtypelist}{%
		\newenvironment{sub#1}%
		{\caption@withoptargs\subcaption@minipage}%
		{\endminipage}}%
	\newcommand*\subcaption@minipage[2]{%
		\minipage#1{#2}%
		\setcaptionsubtype\relax}
	\makeatother
}

\usepackage{url}
\usepackage{afterpage}

\usepackage{tablefootnote}

\usepackage{tikz}
\usetikzlibrary{calc,positioning}

\usepackage[lined,boxed,commentsnumbered,linesnumbered]{algorithm2e}

\newcommand{\institute}[1]{\let\and\\\date{#1}}
\newcommand{\inst}[1]{$^{#1}$}
\newcommand{\keywords}[1]{\par\noindent{\bf KEYWORDS:} #1}
\usepackage{amsthm}
\newtheorem{definition}{Definition}
\newtheorem{proposition}{Proposition}
\newtheorem{example}{Example}
\newtheorem{theorem}{Theorem}
\newtheorem{corollary}{Corollary}

\newcommand{\putBibliography}{
	\bibliography{/home/leshyk/PhD/library.bib}
}
\SetAlCapSkip{10pt}

\newcommand{\PS}{\mathbb{P}} 
\newcommand{\K}{\mathbb{K}}
\newcommand{\Lat}{\mathfrak{L}}
\newcommand{\rc}{\mathbb{R}}
\newcommand{\set}[1]{\{#1\}}

\newcommand{\tikznamedpicture}[3][0]{
	\newcommand{#2}[#1]{ \begin{tikzpicture}#3\end{tikzpicture} }
}

\tikznamedpicture{\putProjContrexample}{[
	every node/.style={draw,circle},
	every edge/.style={draw,very thick},
	node distance= 0.5cm and 0.5cm
]
	\node(bot){$\bot$}
		edge[->,dashed,loop, out=-45,in=-90,looseness=6](bot);
	\node(z)[above=of bot]{$Z$}
		edge (bot)
		edge [->,dashed, in=135,out=-135] (bot);
	\node(x) [above left=of z] {$X$}
		edge (z)
		edge[->,dashed,loop,out=-135,in=-90,looseness=6](x);
	\node(y) [above right=of z] {$Y$}
		edge (z)
		edge[->,dashed,loop,out=-45,in=-90,looseness=6](y);
}

\tikznamedpicture{\putProjIncresingContext}{[
	every node/.style={draw,circle},
	every edge/.style={draw,very thick},
	node distance= 10mm and 10mm
]
	\node(bot){$\bot$}
		edge[->,dashed,loop, out=-75,in=-105,looseness=8](bot);
	\node(b)[above=of bot] {$\set{b}$}
		edge (bot)
		edge[->,dashed,loop, out=-30,in=-60,looseness=7](b);
	\node(a)[left=of b] {$\set{a}$}
		edge (bot)
		edge[->,dashed,out=-90,in=-180](bot);
	\node(c)[right=of b] {$\set{c}$}
		edge (bot)
		edge[->,dashed,loop, out=-30,in=-60,looseness=7](c);
	\node(ab)[above=of a] {$\set{a,b}$}
		edge (a)
		edge (b)
		edge[->,dashed,loop, out=-150,in=-120,looseness=6](ab);
	\node(ac)[above=of b] {$\set{a,c}$}
		edge (a)
		edge (c)
		edge[->,dashed,loop, out=-75,in=-105,looseness=6](ac);
	\node(bc)[above=of c] {$\set{b,c}$}
		edge (b)
		edge (c)
		edge[->,dashed,loop, out=-30,in=-60,looseness=6](bc);
	
	\node(g1)[above right=of ab,draw=none]{$g_1$}
		edge [|->,dotted,out = -90, in=100] (ab);
	\node(g2)[above right=of ac,draw=none]{$g_2$}
		edge [|->,dotted,out = -90, in=100] (ac);
	\node(g3)[above right=of bc,draw=none]{$g_3$}
		edge [|->,dotted,out = -90, in=100] (bc);
}
\newcommand{\concept}[2]{$\left(#1;#2\right)$}

\tikznamedpicture{\putIPSLattice}{[
	every node/.style={draw,rectangle,font={\tiny}},
	node distance= 0.5cm and 0.8cm
	]
	\node(bottom){\concept{\emptyset}{\top}};
	\node(c2)[above=of bottom] {\concept{\set{g_2}}{\left<[2,2];[2,2]\right>}}
		edge (bottom);
	\node(c1)[above left=of bottom] {\concept{\set{g_1}}{\left<[1,1];[1,1]\right>}}
		edge (bottom);
	\node(c3)[above right=of bottom] {\concept{\set{g_3}}{\left<[3,3];[2,2]\right>}}
		edge (bottom);
	\node(c12)[above=of c1,xshift=10mm] {\concept{\set{g_1,g_2}}{\left<[1,2];[1,2]\right>}}
		edge (c1)
		edge (c2);
	\node(c23)[above=of c3,xshift=-10mm] {\concept{\set{g_2,g_3}}{\left<[2,3];[2,2]\right>}}
		edge (c2)
		edge (c3);
	\node(c123)[above=of $(c12)!0.5!(c23)$] {\concept{\set{g_1,g_2,g_3}}{\left<[1,3];[1,2]\right>}}
		edge (c12)
		edge (c23);
}
\tikznamedpicture{\putIPSRCLattice}{[
	every node/.style={draw,rectangle,font={}},
	node distance= 0.5cm and 0.2cm
	]
	\node(bottom){\concept{\emptyset}{\set{a_1,a_2,a_3,a_4,a_5}}};
	\node(c2)[above=of bottom] {\concept{\set{g_2}}{\set{a_3,a_4,a_5}}}
		edge (bottom);
	\node(c1)[above left=of bottom] {\concept{\set{g_1}}{\set{a_1,a_3,a_5}}}
		edge (bottom);
	\node(c3)[above right=of bottom] {\concept{\set{g_3}}{\set{a_2,a_4,a_5}}}
		edge (bottom);
	\node(c12)[above=of c1,xshift=10mm] {\concept{\set{g_1,g_2}}{\set{a_3,a_5}}}
		edge (c1)
		edge (c2);
	\node(c23)[above=of c3,xshift=-10mm] {\concept{\set{g_2,g_3}}{\set{a_4,a_5}}}
		edge (c2)
		edge (c3);
	\node(c123)[above=of $(c12)!0.5!(c23)$] {\concept{\set{g_1,g_2,g_3}}{\set{a_5}}}
		edge (c12)
		edge (c23);
}
\tikznamedpicture{\putFCALattice}{[
	every node/.style={draw,rectangle,font={\bf}},
	node distance= 0.7cm and 0.5cm
	]
	\node(bottom){\concept{}{\set{m_1\+m_2\+m_3\+m_4}}};
	\node(g2g4)[above=of bottom] {\concept{g_2\+g_4}{\set{m_3\+m_4}}}
		edge (bottom);
	\node(g1)[left=of g2g4] {\concept{\set{g_1}}{\set{m_1\+m_4}}}
		edge (bottom);
	\node(g3) [right=of g2g4] {\concept{\set{g_3}}{\set{m_2}}}
		edge (bottom);
	\node (g1g2g4) [above=of $(g1)!0.5!(g2g4)$] {\concept{\set{g_1\+g_2\+g_4}}{\set{m_4}}}
		edge(g1)
		edge(g2g4);
	\node(top) [above=of g1g2g4 -| g2g4] {\concept{\set{g_1\+g_3\+g_2\+g_4}}{}}
		edge(g1g2g4)
		edge(g3);
}

\begin{document}

\title{
	Revisiting Pattern Structure Projections\footnote{The final publication is available at \href{http://link.springer.com/chapter/10.1007\%2F978-3-319-19545-2_13}{link.springer.com}}
}

\author{
	Aleksey Buzmakov\inst{1,2} \and Sergei O. Kuznetsov\inst{2} \and Amedeo Napoli\inst{1}
}

\institute{
	\inst{1}LORIA (CNRS -- Inria NGE -- U. de Lorraine), Vandœuvre-lès-Nancy, France
	\\
	\inst{2}National Research University Higher School of Economics, Moscow, Russia
	\\
	aleksey.buzmakov@inria.fr, skuznetsov@hse.ru, amedeo.napoli@loria.fr
}

%
\maketitle
%
\setcounter{footnote}{0}

\begin{abstract}
	Formal concept analysis (FCA) is a well-founded method for data analysis and has many applications in data mining. Pattern structures is an extension of FCA for dealing with complex data such as sequences or graphs. However the computational complexity of computing with pattern structures is high and projections of pattern structures were introduced for simplifying computation. In this paper we introduce o-projections of pattern structures, a generalization of projections which defines a wider class of projections preserving the properties of the original approach. Moreover, we show that o-projections form a semilattice and we discuss the correspondence between o-projections and the representation contexts of o-projected pattern structures.
\end{abstract}

\section{Introduction}
A significant part of recorded data represents phenomena in a structured way, e.g., a molecule is better represented as a labeled graph than as a set of attributes. Pattern structures are an extension of FCA for dealing with such kind of data~\cite{Ganter1999,Ganter2001,Ganter2004}.
Such a pattern structure is defined by a set of objects, a set of descriptions associated with the set of objects, and a similarity operation on descriptions, matching a pair of descriptions to their common part. For instance, the set of objects can contain molecule names, the set of descriptions contains fragments of molecules, and the similarity operation taking two sets of graphs to a set of maximal common subgraphs. The similarity operation is a semilattice operation on the set of descriptions. It allows one to deal with data (objects and their descriptions) in a similar way as one deals with objects and their intents in standard FCA. Such kind of formalization allows one to describe many types of data, however processing can be computationally very demanding. For example, pattern structures on sets of graphs~\cite{Ganter2001,Ganter2004,Kuznetsov2005} is based on the operation of finding maximal common subgraphs for a set of graphs, which is \#P-hard.

To deal with this complexity and to have a possibility to process most of the data, projections of pattern structures were introduced~\cite{Ganter2001}. Projections are special mathematical functions on the set of descriptions that simplify the descriptions of objects. This approach reduces the number of concepts in the pattern lattice corresponding to a pattern structure. However, it does not impact the computational worst-case complexity of the similarity operation. Moreover, it cannot remove concepts of special kinds from the ``middle'' of the semilattice which can be important in some practical cases, e.g., concepts containing too small graphs can be considered useless but they cannot be removed with projections. For example, in~\cite{Buzmakov2013a} concepts having intents that include short sequences of patient hospitalisations have little sense. Hence, short sequences could be ``removed'' from the intent, but the descriptions of objects, i.e., patients, usually include only one long sequence and should not be changed.

In this paper we introduce \textit{o-projections} of pattern structures, a generalization of projections of pattern structures, that allow one to reduce the computational complexity of similarity operations. They also allow one to remove certain kinds of descriptions in the ``middle'' of the semilattice while the descriptions of the objects can be preserved. By introducing o-projections of pattern structures, we correct also some formal problems of projections of pattern structures, which will be discussed later.

The main difference between o-projections and projections is that in o-projected pattern structures we modify the semilattice of descriptions, while in the case of projected pattern structures we can modify only the descriptions of single objects. It should be noticed that most of the properties of projections are valid for o-projections. However, the relation between representation contexts, a reduction from pattern structures to FCA, and projections is quite different from the relation between representation contexts and o-projections. The introduction and study of this difference is one of the main contributions of this work. In addition we have discovered the fact that the set of o-projections of a pattern structure forms a semilattice. From a practical point of view it allows one to apply a set of independent o-projections, e.g., o-projections obtained from several experts, to a pattern structure.

This work further develops the methodology introduced in~\cite{Buzmakov2013a}, where it was applied for the analysis of sequential pattern structures by introducing projections that remove irrelevant concepts.

The rest of the paper is organized as follows. In Section~\ref{sect:PS} we introduce the definitions of a pattern structure, representation context of a pattern structure, and discuss how one can compute with pattern structures along the lines of FCA.  Section~\ref{sect:projections} introduces projections and o-projections of a pattern structure, defines the partial order on o-projections and shows that this order is a semilattice. At the end of this section the relation between o-projections and representation contexts of o-projected pattern structure is discussed. Finally, we conclude the paper and discuss furture work.

\section{Pattern Structures}\label{sect:PS}
In FCA a formal context $(G,M,I)$, where $G$ is a set of objects, $M$ is a set of attributes, and $I \subseteq G \times M$ is a binary relation between $G$ and $M$, is taken to a concept lattice
$\Lat (G,M,I)$~\cite{Ganter1999}. For non-binary data, such as sequences or graphs, lattices  can be constructed in the same way using  pattern structures~\cite{Ganter2001}.
\begin{definition}\label{def:PS}
	A pattern structure $\PS$ is a triple $(G,(D,\sqcap),\delta)$, where $G, D$ are sets, called the set of objects and the set of descriptions, and $\delta: G \rightarrow D$ maps an object to a description. Respectively, $(D,\sqcap)$ is a meet-semilattice on $D$ w.r.t. $\sqcap$, called similarity operation such that $\delta(G):=\set{\delta(g) \mid g \in G}$ generates a complete subsemilattice $(D_\delta,\sqcap)$ of $(D,\sqcap)$.
\end{definition}

\newcommand{\sqcapb}{\sqcap}

 For illustration, let us represent standard FCA in terms of pattern structures. The set of objects $G$ is preserved, the semilattice of descriptions is $(\wp(M),\cap)$, where $\wp(M)$ denotes the powerset of the set of attributes $M$, a description is a subset of attributes and $\cap$ is the set-theoretic intersection. If $x=\set{a\+b\+c}$ and $y=\set{a\+c\+d}$ then $x \sqcapb y = x \cap y = \set{a\+c}$, and  $\delta: G \rightarrow \wp(M)$ is given by $\delta(g)=\{m \in M \mid (g,m) \in I\}$.

 Note that Definition~\ref{def:PS} has an important partial case where $(D,\sqcap)$ is a complete meet-semilattice. In this case the semilattice $(D_\delta,\sqcap)$ is necessarily complete.
First,  in practical applications one often needs finite lattices, which are always complete. Second, in many practical cases one can easily extend an incomplete semilattice to a complete one by introducing some extra elements. For example, given an incomplete semilattice w.r.t containment order on the interval $(a,b)$, one can add $a$ and $b$ to obtain the interval $[a,b]$, which is a complete semilattice.
In this paper some of the statements hold only for the partial case of $(D,\sqcap)$ being a complete meet-semilattice.

The Galois connection for a pattern structure $(G,(D,\sqcap),\delta)$, relating sets of objects and descriptions, is defined as follows:
\begin{align*}
	A^\diamond &:= \underset{g \in A}{\bigsqcap}\delta(g), &\text{for } A \subseteq G\\
	d^\diamond &:= \{g \in G \mid d \sqsubseteq \delta(g)\}, &\text{for } d \in D
\end{align*}

Given a subset of objects $A$, $A^\diamond$ returns the description which is common to all objects in $A$. Given a description $d$, $d^\diamond$ is the set of all objects whose description subsumes $d$.
The natural partial order (or subsumption order between descriptions) $\sqsubseteq$ on $D$  is defined w.r.t. the similarity operation $\sqcap$:
$c \sqsubseteq d \Leftrightarrow c \sqcap d = c$ (in this case we say that $c$ is subsumed by $d$). In the case of standard FCA the natural partial order corresponds to the set-theoretical inclusion order, i.e., for two sets of attributes $x$ and $y$ $x\sqsubseteq y \Leftrightarrow x \subseteq y$.

\begin{definition}
	A pattern concept of a pattern structure $(G,(D,\sqcap),\delta)$ is a pair $(A,d)$, where $A \subseteq G$ and $d \in D$ such that $A^\diamond = d$ and $d^\diamond = A$; $A$ is called the pattern extent and $d$ is called the pattern intent.
\end{definition}

As in standard FCA, a pattern concept corresponds to the maximal set of objects $A$ whose description subsumes the description $d$, where $d$ is the maximal common description of objects in $A$.
The set of all pattern concepts is partially ordered w.r.t. inclusion of extents or, dually, w.r.t. subsumption of pattern intents within a concept lattice, these two antiisomorphic orders making a lattice, called pattern lattice.

\subsection{Running Example}\label{sect:example}
%

The authors of~\cite{Kaytoue2011} have used interval pattern structures for gene expression analysis. Let us consider an example of such pattern structures. In Figure~\ref{fig:ex-context} an interval context is shown. It has three objects and two attributes. Every attribute shows the interval of values the attribute can have. If we have two objects, then a numerical attribute can have all values from the interval of this attribute in the first object and from the interval of this attribute of the second object. Consequently, the similarity between two intervals can be defined as a convex hull of the intervals, i.e. $[a,b] \sqcap [c,d]=[\min(a,c),\max(b,d)]$. Then, given two tuples of intervals, the similarity between these tuples is computed as a component-wise similarity between intervals.

\begin{figure}[t]
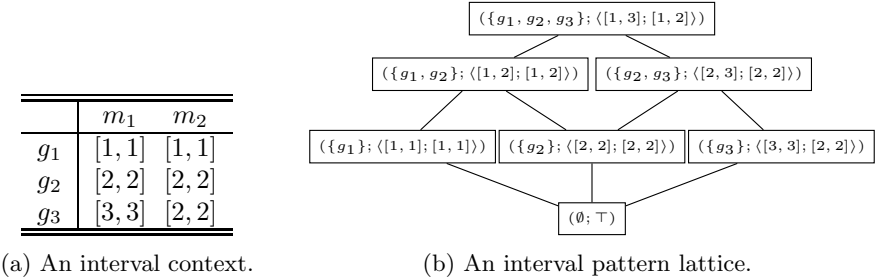

	\centering
	\begin{subfigure}[b]{0.34\columnwidth}
		\centering
		\begin{tabular}{r|c@{~~}c}
			\hline\hline
			& $m_1$ & $m_2$ \\
			\hline
			$g_1$
			& $[1,1]$ & $[1,1]$ \\
			$g_2$
			& $[2,2]$ & $[2,2]$ \\
			$g_3$
			& $[3,3]$ & $[2,2]$ \\
			\hline\hline
		\end{tabular}
		\caption{An interval context.}
		\label{fig:ex-context}
	\end{subfigure}
	\begin{subfigure}[b]{0.65\columnwidth}
		\centering
		\resizebox{\columnwidth}{!}{
			\putIPSLattice
		}
		\caption{An interval pattern lattice.}
		\label{fig:ex-lattice}
	\end{subfigure}
	\caption{An interval pattern structure and the corresponding lattice.}
	\label{fig:ex-PS}
\end{figure}

In this example, we have the pattern structure $(G,(D,\sqcap),\delta)$, where $G=\set{g_1,g_2,g_3}$, the set $D$ is the set of all possible interval pairs with the similarity operation described above, and $\delta$ is given by the context in Figure~\ref{fig:ex-context}, i.e., $\delta(g_1) = \left<[1,1];[1,1]\right>$ and $\delta(g_1) \sqcap \delta(g_2)=\left<[1,2];[1,2]\right>$.

Figure~\ref{fig:ex-lattice} shows the pattern lattice of the interval context in Figure~\ref{fig:ex-context}. One can check that the extents and the intents in this lattice are connected by means of the Galois connection given above. The partial order in the semilattice of intervals is given by ``the smaller the interval, the larger the description with this interval", i.e., the former description gives more certainty about the values than the latter.

\subsection{Representation Context of a Pattern Structure}

Note that any pattern structure can be represented by a formal context with the concept lattice isomorphic to the lattice of the pattern structure. Below we introduce a representation context of a pattern structure and its properties in the line of~\cite{Ganter2001}.

Given a pattern structure $(G,(D,\sqcap),\delta)$, we denote by $D_\delta \subseteq D$ the set of all intents of the concept lattice, i.e., $D_\delta=\set{d \in D \mid (\exists X \subseteq G)\bigsqcap\limits_{g \in X}\delta(g) = d}$. Since $(D_\delta, \sqcap)$ is a complete subsemilattice of $(D,\sqcap)$, for $X \subseteq D$ a join operation $\sqcup$ can be defined as follows: \[
	\bigsqcup X = \bigsqcap \set{d \in D_\delta \mid (\forall x \in X) x \sqsubseteq d}.
\]Given this join operation, $(D_\delta,\sqcap,\sqcup)$ is a complete lattice. We say that a set $M\subseteq D$ is $\sqcup$-dense for $(D_\delta,\sqcap)$ if every element in $D_\delta$ is of the form $\sqcup X$ for some $X \subseteq M$. For example, $M=D_\delta$ is always $\sqcup$-dense for $D_\delta$.

\begin{definition}
	Given a pattern structure $\PS=(G,(D,\sqcap),\delta)$ and a set $M \subseteq D$ $\sqcup$-dense in $D_\delta$, a formal context $(G,M,I)$ is called the representation context of $\PS$, if $I$ is given by $I=\set{(g,m) \in G \times M \mid m \sqsubseteq \delta(g)}$. The representation context of $\PS$ is denoted by $\rc(\PS)$.
\end{definition}

The next theorem establishes a bijection between the pattern concepts in the lattice of pattern structure $\PS$ and the concepts in the lattice of the representation context $\rc(\PS)$. Here, the ideal of element $d \in D$ is denoted by $\downarrow d=\set{e \in D \mid e \sqsubseteq d}$.

\begin{theorem}[Theorem~1 from~\cite{Ganter2001}] \label{thm:representation-context}
	Let $\PS=(G,(D,\sqcap),\delta)$ be a pattern structure and let $\rc(\PS)=(G,M,I)$ be a representation context of $\PS$. Then for any $A \subseteq G$, $B \subseteq M$, and $d \in D$ the following conditions are equivalent:
	\begin{enumerate}
		\item
			$(A,d)$ is a pattern concept of $\PS$ and $B=\downarrow d \cap M$.
		\item
			$(A,B)$ is a formal concept of $\rc(\PS)$ and $d = \bigsqcup B$.
	\end{enumerate}
\end{theorem}

\begin{example}
	A representation context for the pattern structure given in Figure~\ref{fig:ex-PS} can be given by the set $M$ where every element $m \in M$ is of the form $\left<[-\infty,a];[-\infty,+\infty]\right>$ or $\left<[-\infty,+\infty];[b,+\infty]\right>$, and $a,b\in \{1,2,3\}$.
	
	In fact, the element $\left<[-\infty,+\infty];[a,+\infty]\right>$ corresponds to the attribute '$m_2 \geq a$' in the case of the interordinal scaling~\cite{Ganter1999} of numerical data. Another representation context can be constructed from the intents of join-irreducible concepts of the lattice in Figure~\ref{fig:ex-lattice}. These two representation contexts of the pattern structure related to Figure~\ref{fig:ex-PS} are shown in Figures~\ref{fig:ex-RC-interordinal}~and~\ref{fig:ex-RC-another}. It can be seen that the resulting lattices, e.g., the lattice in Figure~\ref{fig:ex-RC-lattice}, are isomorphic to the lattice in Figure~\ref{fig:ex-lattice}.
\end{example}

\begin{figure}[t]
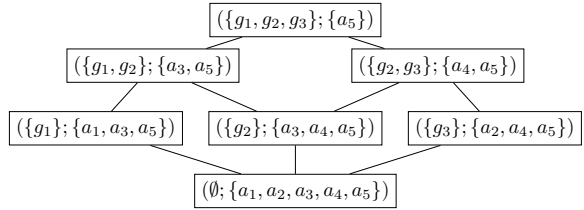

	\centering
	\begin{subfigure}[b]{0.6\columnwidth}
		\centering
		\begin{tabular}{l|cccc|cc|c}
			\hline\hline
			& \rotatebox{90}{\scriptsize$\left<[3,+\infty];[-\infty,+\infty]\right>$}
			& \rotatebox{90}{\scriptsize$\left<[2,+\infty];[-\infty,+\infty]\right>$}
			& \rotatebox{90}{\scriptsize$\left<[-\infty,1];[-\infty,+\infty]\right>$}
			& \rotatebox{90}{\scriptsize$\left<[-\infty,2];[-\infty,+\infty]\right>$}

			& \rotatebox{90}{\scriptsize$\left<[-\infty,+\infty];[2,\infty]\right>$}
			& \rotatebox{90}{\scriptsize$\left<[-\infty,+\infty];[-\infty,1]\right>$}

			& \rotatebox{90}{\scriptsize$\left<[1,3];[1,2]\right>$}
			\\
			\hline
			& \scriptsize $m_1 \geq 3$ & \scriptsize $m_1 \geq 2$ & \scriptsize $m_1 \leq 1$ & \scriptsize $m_1 \leq 2$
			& \scriptsize $m_2 \geq 2$ & \scriptsize $m_2 \leq 1$ 
			& \\
			\hline
			$g_1$
			&   &   & x & x &   & x & x \\
			$g_2$
			&   & x &   & x & x &   & x \\
			$g_3$
			& x & x &   &   & x &   & x \\
			\hline\hline
		\end{tabular}
		\caption{Representation context corresponding to interordinal scaling.}
		\label{fig:ex-RC-interordinal}
	\end{subfigure}
	\hfill
	\begin{subfigure}[b]{0.35\columnwidth}
		\centering
		\begin{tabular}{l|ccccc}
			\hline\hline
			& \rotatebox{90}{\scriptsize$\left<[1,1];[1,1]\right>$}
			& \rotatebox{90}{\scriptsize$\left<[3,3];[2,2]\right>$}
			& \rotatebox{90}{\scriptsize$\left<[1,2];[1,2]\right>$}
			& \rotatebox{90}{\scriptsize$\left<[2,3];[2,2]\right>$}
			& \rotatebox{90}{\scriptsize$\left<[1,3];[1,2]\right>$}
			\\ \hline
			& $a_1$ & $a_2$ & $a_3$ & $a_4$ & $a_5$  \\
			\hline
			$g_1$
			& x &   & x &   & x \\
			$g_2$
			&   &   & x & x & x \\
			$g_3$
			&   & x &   & x & x \\
			\hline\hline
		\end{tabular}
		\caption{Another possible representation context.}
		\label{fig:ex-RC-another}
	\end{subfigure}
	\begin{subfigure}[b]{0.65\columnwidth}
		\centering
		\resizebox{\columnwidth}{!}{
			\putIPSRCLattice
		}
		\caption{A concept lattice for the context if Figure~\ref{fig:ex-RC-another}.}
		\label{fig:ex-RC-lattice}
	\end{subfigure}
	\caption{Possible representation contexts for the pattern structure in Figure~\ref{fig:ex-PS} and the concept lattice for the context in Figure~\ref{fig:ex-RC-another}.}
	\label{fig:ex-RCs}
\end{figure}

It should be noticed that in some cases the representation context is hard to compute. For example, in case of  numerical data with the set of all values $W$, to construct representation context, one needs to create  $2|W| +1$ binary attributes, which can be much more than the number of original real-valued attributes. The authors of~\cite{Kaytoue2011} have shown that pattern structures provide more efficient computations than the equivalent approach based on FCA and scaling, which can be considered as a way to build representation context of interval pattern structures, e.g., see Figure~\ref{fig:ex-RC-interordinal}.

In case of graph data the set of attributes of the representation context consists of all subgraphs of the original graph descriptions, which is hard to compute~\cite{Kuznetsov2005}.

\subsection{Computation of Pattern Lattices}
Nearly any algorithm for computing concept lattices from contexts can be adapted to compute pattern lattices from pattern structures. To adapt an algorithm, every set intersection operation on attributes is replaced by the semilattice operation $\sqcap$ on corresponding patterns, and every subset checking is replaced by the semilattice order  $\sqsubseteq$ checking, in particular, all $(\cdot)'$ operations are replaced by $(\cdot)^\diamond$. For example, let us consider a modified version of Close-by-One (\texttt{CbO}) algorithm~\cite{Kuznetsov1993}.

\begin{algorithm}[t]
	\AlgoDisplayBlockMarkers
	\SetAlgoBlockMarkers{}{}
	\SetAlgoNoEnd
	\SetKwProg{Fn}{Function}{}{}
	\SetKwFunction{CbO}{CloseByOne}
	\SetKwFunction{IsCanonicExt}{IsCanonicExtension}
	\SetKwFunction{Save}{SaveConcept}
	
	\Fn{\CbO{$Ext$, $Int$}}{
		\KwData{$\PS=(G,(D,\sqcap),\delta)$, the extent $Ext$ and the intent $Int$ of a concept. }
		\KwResult{All canonical ancestors of $(Ext,Int)$ in the concept lattice.}
		\ForEach{$S\subseteq G$, $S \succ Ext$}{
			$NewInt \longleftarrow \underset{g \in S}{\pmb{\bigsqcap}}\delta(g)$
			\tcc*[r]{$\sqcap$ - the similarity}
			$NewExt \longleftarrow \{g \in G \mid  NewInt \pmb{\sqsubseteq} \delta(g)\}$
			\tcc*[r]{$\sqsubseteq$ - the subsumption}

			\If{\IsCanonicExt{$Ext$, $NewExt$}}{
				\Save{$(NewExt,NewInt)$}\;
				\CbO{NewExt,NewInt}\;
			}
		}
	}
	\tcc{Looking for all concepts of the concept lattice}
	\CbO{$\emptyset$, $\top$}\;
	\caption{The version of the Close-by-One algorithm computing the pattern lattice of a pattern structure $\PS$.  \label{alg:cbo-ps}}
\end{algorithm}

Algorithm~\ref{alg:cbo-ps} shows the listing of the modified part of \texttt{CbO}. Here the canonical extension \texttt{IsCanonicExtension} and canonical order $\succ$ are defined on the set of objects and hence are the same as in~\cite{Kuznetsov1993}. We can see that only lines 3 and 4 are modified. In these lines the set intersection operation and the subset relation checking are replaced by the corresponding operators of a pattern structure.

\section{Revised Projections of Pattern Structures}\label{sect:projections}
Pattern structures are hard to process due to the large number of pattern concepts in the pattern lattice and the algorithmic complexity of the similarity operation $\sqcap$. Projections of pattern structures ``simplify'' to some degree the computation and allow one to work with ``simpler" descriptions. In fact, a projection can be considered as a mapping for pruning descriptions with certain mathematical properties. These properties ensure that a projection of a semilattice is a semilattice and that the concepts of a projected\footnote{
	We use the expression ``a projected pattern structure'' instead of ``a projection of a pattern structure'' to distinguish between projection as an operator $\psi$ and as the result of applying the operator to a pattern structure.
} pattern structure are related to the concepts of the original pattern structure~\cite{Ganter2001}.

In this section we introduce o-projected pattern structures (``o" coming from ``order"), i.e. a revision of projected pattern structures in accordance with~\cite{Buzmakov2013a}. We discuss the properties of o-projected pattern structures and relate them to the projected pattern structures from~\cite{Ganter2001}. The notion of (o-)projected pattern structure is based on a kernel operator (a projection).

\begin{definition}[\cite{Ganter2001}]\label{def:projection}
	A projection $\psi: D \rightarrow D$ is a kernel (interior) operator on the partial order $(D,\sqsubseteq)$, i.e. it is (1)~monotone ($x \sqsubseteq y \Rightarrow \psi(x) \sqsubseteq \psi(y)$), (2)~contractive ($\psi(x) \sqsubseteq x$) and (3)~idempotent ($\psi(\psi(x))=\psi(x)$).
\end{definition}

Given a projection $\psi$ we say that the fixed point of $\psi$ is the set of all elements from $D$ such that they are mapped to themselves by $\psi$. The fixed point of $\psi$ is denoted by $\psi(D)=\{d \in D \mid \psi(d)=d\}$. Note that, if $\psi(d) \neq d$, then there is no other $\tilde{d}$ such that $\psi(\tilde{d}) = d$ because of idempotency. Hence, any element outside the fixed point of the projection $\psi$ is pruned.

\subsection{Definition of Projected Pattern Structures}

Let us first consider the projected pattern structure w.r.t. a projection $\psi$ according to~\cite{Ganter2001}. Given a pattern structure $\PS=(G,(D,\sqcap),\delta)$ and a projection $\psi$ on $D$, the projected pattern structure is defined as $(G,(D,\sqcap),\psi \circ \delta)$. As we can see, a projection only changes the descriptions of the objects but not the underlying semilattice $(D,\sqcap)$. There are two problems with this definition of the projected pattern structures. First, it is necessary to restrict the class of projections given by Definition~\ref{def:projection} in order to ensure the property $\psi(x\sqcap y) = \psi(x)\sqcap \psi(y)$. Second, the complexity of computing $\sqcap$ can be very high, but with this kind of projected patten structures we cannot decrease the algorithmic complexity. Below we discuss these two points.

In~\cite{Ganter2001} (Proposition~1) the following property of the projection operator is discussed: given a semilattice $(D,\sqcap)$ and a projection $\psi$ on $D$, for any two elements $x$ and $y$ from $D$ one has $\psi(x \sqcap y) = \psi(x) \sqcap \psi(y)$. Let us consider the example in Figure~\ref{fig:semilattice-contrexample} with the meet-semilattice $D=\set{x,y,z,\bot}$ given by its diagram and the projection $\psi$ given by the dotted lines. It is easy to see that $\psi(x\sqcap y)=\bot \neq z=\psi(x) \sqcap \psi(y)$. One way of solving this problem is to give additional conditions on projection $\psi$ that would imply the required property. An important example is the following condition: for all $x,y \in D$ if $x < y$ and $\psi(y)=y$, then $\psi(x)=x$. This kind of solution respects the intuition behind the definition of the projected pattern structure in~\cite{Ganter2001}, according to which the initial descriptions of objects are changed, but the similarity operation $\sqcap$ is not changed.

\begin{figure}[t]
	\centering
	\begin{minipage}{0.45\columnwidth}
		\putProjContrexample
	\end{minipage}
	\qquad
	\begin{minipage}{0.45\columnwidth}
		\[
			D=\set{x,y,z,\bot}
		\]

		\vspace{0,3cm}

		\[\begin{aligned}
				\psi: &x \mapsto x, y \mapsto y,\\
				&z \mapsto \bot, \bot \mapsto \bot
		\end{aligned}\]

		\vspace{0,3cm}

		\[\begin{aligned}
				\psi(x \sqcap y) &= \psi(z) = \bot \neq\\
				&\neq z = \psi(x) \sqcap \psi(y)
		\end{aligned}\]

	\end{minipage}
	\caption{Contrexample to Proposition~1 from~\cite{Ganter2001}.}
	\label{fig:semilattice-contrexample}
\end{figure}

Another way of solving the problem above is to generalize the definition of the projected pattern structure, and we proceed in this way in the next section, by allowing to modify the similarity operation on descriptions.


%
\subsection{Definition of o-projected Pattern Structures}

Below we propose a definition of o-projected pattern structures by means of a kernel operator $\psi$. The definition takes into account the problems discussed above.
In the o-projected pattern structure we substitute the semilattice of descriptions by its suborder (the letter ``o'' comes from ``order'') with another similarity operation, which can be different from the initial one.

Let us first note that, given a meet-semilattice $D$ and a kernel operator $\psi$, the fixed point $\psi(D)$ is a semilattice w.r.t. to the natural order on $D$.

\newcommand{\sqcappsi}{\sqcap_\psi}
\newcommand{\bigsqcappsi}{{\bigsqcap}_\psi}

\begin{theorem} \label{thm:psi-sqcap=sqcap-psi}
	Given a semilattice $(D,\sqcap)$ and a kernel operator $\psi$, the fixed point $(\psi(D),\sqcappsi)$ is a semilattice w.r.t. the natural order on $(D,\sqcap)$, i.e., $d_1 \sqsubseteq d_2 \Leftrightarrow d_1 \sqcap d_2=d_1$. If $\bigsqcap X$ exists for a set $X \subseteq D$, then $\underset{x \in X}{\bigsqcappsi} \psi(x)$ exists and is given by
	\begin{equation} \label{eq:psi-sqcap=sqcap-psi}
		\underset{x \in X}{\bigsqcappsi} \psi(x) = \psi(\underset{x \in X}{\bigsqcap}x)
	\end{equation}
\end{theorem}
\begin{proof}
	Let us denote $d=\underset{x \in X}{\bigsqcap}x$. Since $(\forall x \in X) d \sqsubseteq x$, one has $(\forall x \in X) \psi(d) \sqsubseteq \psi(x)$. Let us show that for any $p \in \psi(D)$, i.e. $\psi(p)=p$ such that $(\forall x \in X) p \sqsubseteq \psi(x)$, we have $p \sqsubseteq \psi(d)$, i.e., that $\psi(d) = \underset{x \in X}{\bigsqcappsi}\psi(x)$.

Since $(\forall x \in X) p \sqsubseteq \psi(x)$ then $(\forall x \in X) p \sqsubseteq x$. Since $d=\underset{x \in X}{\bigsqcap}x$, one has $p \sqsubseteq d$. Thus, $p=\psi(p) \sqsubseteq \psi(d)$ and $\psi(d)$ is the minimum of the set $\psi(X)$, i.e. $\psi(D)$ is a semilattice and the Eq.~(\ref{eq:psi-sqcap=sqcap-psi}) holds.
\end{proof}
\begin{corollary}\label{cor:subsemilattice-image}
	Given a complete subsemilattice $\tilde{D}$ of $(D,\sqcap)$ and a kernel operator $\psi$ on $D$, the image of $\tilde{D}$ is a complete subsemilattice $\psi(\tilde{D})$ of the fixed point $(\psi(D),\sqcappsi)$.
\end{corollary}


Since according to Theorem~\ref{thm:psi-sqcap=sqcap-psi} $\psi(D)$ is a semilattice and according to Corollary~\ref{cor:subsemilattice-image} $\psi(D_\delta)$ is a complete semilattice, we can define an o-projected pattern structure as a pattern structure with $\psi(D)$ as a semilattice.

\begin{definition}\label{def:projected-PS}
 Given a pattern structure $\PS=(G,(D,\sqcap),\delta)$ and a kernel operator $\psi$ on $D$, the o-projected pattern structure $\psi(\PS)$ is a pattern structure $(G,(\psi(D),\sqcappsi),\psi \circ \delta)$, where $\psi(D)=\{d\in D \mid \psi(d) = d\}$ and $\forall x,y \in D, x \sqcappsi y := \psi( x \sqcap y )$.
\end{definition}

In the o-projected pattern structure the kernel operator $\psi$ modifies not only the descriptions of the objects, but also the semilattice operation, i.e., the semilattice $(\psi(D), \sqcap_{\psi})$ is not necessarily a subsemilattice of $(D,\sqcap)$ and so it is not always true that $x \sqcap y =x \sqcappsi y$ in $D$.

\begin{example}\label{ex:aggregated-length}
	Let us define an o-projection for the interval pattern structure from Subsection~\ref{sect:example}. Let us suppose that the {\it aggregated size of a pattern}, i.e., the sum of the lengths of the intervals in the pattern, should be less than 2. First, we should define the corresponding kernel operator $\psi:D\rightarrow D$. Thus, if an aggregated length of a pattern $p$ is less than 2, then $\psi(p):=p$, otherwise $\psi(p):=\bot=\left<[-\infty,+\infty];[-\infty,+\infty]\right>$. For instance, $\psi(\left<[1,1];[1,1]\right>)=\left<[1,1];[1,1]\right>$, while $\psi(\left<[1,2];[1,2]\right>)=\left<[-\infty,+\infty];[-\infty,+\infty]\right>$, because it has two intervals of length 1, i.e., the aggregated size is equal to 2.
	
	Let us consider the o-projected interval pattern structure $(G,(\psi(D),\sqcappsi),\psi\circ\delta)$. It is clear that $\psi\circ\delta=\delta$, thus this o-projected interval pattern structure cannot be expressed as a projected pattern structure.
\end{example}

The concepts of a pattern structure and a projected pattern structure are connected through Proposition~\ref{prop:proj-and-concepts}. This proposition can be found in \cite{Ganter2001}, but thanks to Theorem~\ref{thm:psi-sqcap=sqcap-psi}, it is also valid in our case.
\begin{proposition}\label{prop:proj-and-concepts}
	Given a pattern structure $\PS=(G,(D,\sqcap),\delta)$ and a kernel operator $\psi$ on $D$:
	\begin{enumerate}
		\item
			if $A$ is an extent in $\psi(\PS)$, then $A$ is also an extent in $\PS$.
		\item
			if $d$ is an intent in $\PS$, then $\psi(d)$ is also an intent in $\psi(\PS)$.
	\end{enumerate}
\end{proposition}

It is easy to see that the other propositions from~\cite{Ganter2001} concerning projected pattern structures hold for the o-projected pattern structures as well. Below we cite Proposition 3 from~\cite{Ganter2001} that relates implications in a pattern structure and those in an o-projected pattern structure. We skip the propositions related to supervised classification with projected pattern structures by means of hypotheses, because it is out of the scope of this paper. However, they are valid in the case of o-projected pattern structures and can be proven with the help of Theorem~\ref{thm:psi-sqcap=sqcap-psi}.

\begin{proposition}[Proposition~3~from~\cite{Ganter2001}]
	Let $a,b \in D$. If $\psi(a) \rightarrow \psi(b)$ and $\psi(b)=b$ then $a \rightarrow b$, where $x \rightarrow y \Leftrightarrow$ for all $g \in G$ $(x \sqsubseteq \delta(g)$ implies $y \sqsubseteq \delta(g))$
\end{proposition}

\subsection{Order of Projections}

In this subsection we limit ourselves to the practically important case when a set of descriptions is a complete semilattice.
We can consider projections as a means of description pruning in $(D,\sqcap)$. Indeed, given a semilattice $(D,\sqcap)$ and a projection $\psi$ on this semilattice, the set $D$ can be divided into two sets $D=\{d \in D \mid \psi(d)=d\} \cup \{d \in D \mid \psi(d) \neq d\}$, i.e., the fixed point of $\psi$ and the rest. It can be seen that the intents of the o-projected pattern structure $\psi((G,(D,\sqcap),\delta))$ are in the fixed point of $\psi$, i.e., all elements of the form $\psi(d)\neq d$ are discarded. We recall that by $\psi(D)=\{d \in D \mid \psi(d)=d)\}$ we denote the fixed point of $\psi$.
\textit{But under which condition do we have that for any $D_1 \subset D_2$ there is a projection $\psi$ of $D_2$ such that $\psi(D_2)=D_1$?} The following theorem gives necessary and sufficient conditions for such a property.

\begin{theorem}\label{thm:suborder-of-a-semilattice}
	Given a complete semilattice $(D,\wedge)$, with the natural order $\leq$, and $D_s\subseteq D$, there is a projection $\psi:D\rightarrow D$ such that $\psi(D)=D_s$, if and only if $\bot \in D_s$ and for any $X \subseteq D_s \subseteq D$, one has $\bigvee X \in D_s$, where $\bot:=\bigwedge D$ and $\bigvee X=\bigwedge\{d \in D \mid (\forall x \in X) d \geq x\}$.
\end{theorem}
\begin{proof}
	\begin{enumerate}
		\item
			Given a projection $\psi$ such  that $\psi(D)=D_s$, $\bot \in D_s$ because of contractivity of $\psi$, i.e., $\psi(\bot)=\bot$.
			Let us suppose that there is a set $X \subseteq D_s$, i.e., $(\forall x \in X)\psi(x)=x$ such that $\psi(\bigvee X) \neq \bigvee X$. Then, $(\forall x \in X)(x < \bigvee X \underset{\text{monotonicity}}{\Rightarrow} x \leq \psi(\bigvee X) \underset{\text{contractivity}}{<} \bigvee X)$.  It is a contradiction, since $\bigvee X$ is the supremum of $X$. Hence for any $X \subseteq D_s$ we have $\psi(\bigvee X)=\bigvee X$.

		\item
			Given $D_s \subseteq D$ such that $\bot \in D_s$ and for any $X \subseteq D_s$, one has $\bigvee X \in D_s$, let us construct the corresponding projection $\psi$. First, $\psi(d \in D_s):=d$ and for all $d \in D \setminus D_s$ we should have $\psi(d) \neq d$.
			For an element $d \in D \setminus D_s$, let us consider the set $S_d=\{x \in D_s \mid x < d\}$, which is not an empty set since $\bot \in D_s$. We know that $\bigvee S_d \in D_s$ and by definition of $\bigvee$ we have $\bigvee S_d < d$. Then we set $\psi(d):=\bigvee S_d$.

			Let us show that the function $\psi$ is a projection of $D$. Idempotency and contractivity are satisfied by the construction of $\psi$. Let us check monotonicity. Let us take any $a,b \in D$ such that $a > b$. Then, if $\psi(a)=a$, then $\psi(a)=a > b \geq \psi(b)$, i.e., the monotonicity holds. If $\psi(a) \neq a$, then $\psi(a) = \bigvee S_a$ by construction. Hence, if $\psi(b)=b$, then $b \in S_a$, i.e., $\psi(a) \geq \psi(b)$. Finally, if $\psi(b) \neq b$, then $S_b \subseteq S_a$, because if $d \in S_b$, i.e., $d < b$, then $d < b < a$, i.e. $d \in S_a$. In this case, $\psi(a)=\bigvee S_a \geq \bigvee S_b=\psi(b)$.
	\end{enumerate}
\end{proof}
\begin{corollary}
	Given a complete semilattice $(D,\wedge)$, with the natural order $\leq$, and a subset $D_s\subseteq D$ such that $\bot \in D_s$ and for any $X \subseteq D_s$, one has $\bigvee X \in D_s$, the poset $(D_s,\leq)$ is a complete semilattice.
\end{corollary}
\begin{proof}
	According to Theorem~\ref{thm:suborder-of-a-semilattice} there is a projection $\psi:D \rightarrow D$ such that $\psi(D)=D_s$. Then, according to Theorem~\ref{thm:psi-sqcap=sqcap-psi} $D_s$ is a semilattice.
\end{proof}

Since a projection of $D$ can be considered as a mapping with the fixed point $\psi(D)$, we can introduce an order w.r.t. this fixed point.

\begin{definition}\label{def:projections-order-fixed-point}
	Given a complete semilattice $(D,\sqcap)$ and two projections $\psi_1$ and $\psi_2$ in $D$, we say that $\psi_1 \leq \psi_2$ if $\psi_1(D) \subseteq \psi_2(D)$.	
\end{definition}

However in some cases, it is more convenient to order projections w.r.t. a superposition of projections or their ``generality''.

\begin{definition}
	\label{def:projections-order-superposition}
	Given a complete semilattice $(D, \sqcap)$ and two projections $\psi_1$ and $\psi_2$ in $D$, we say that $\psi_1 \leq \psi_2$ if there is a projection $\psi:\psi_2(D)\rightarrow\psi_2(D)$ such that $\psi_1 = \psi \circ \psi_2$.
\end{definition}

It can be seen that these two definitions yield the same ordering.

\begin{proposition}
	Definitions~\ref{def:projections-order-fixed-point}~and~\ref{def:projections-order-superposition} are equivalent.
\end{proposition}
\begin{proof}
	\begin{enumerate}
		\item
			Let $\psi_1=\psi\circ\psi_2$. Since $\psi$ is a projection in $\psi_2(D)$, then $\psi_1(D)=\psi(\psi_2(D)) \subseteq \psi_2(D)$.
		\item
			Let $\psi_1(D) \subseteq \psi_2(D)$. Let us denote by $(\cdot)_1$ and $(\cdot)_2$ the operations in $({\psi_1}(D),\sqcap_{\psi_1})$ and $({\psi_2}(D),\sqcap_{\psi_2})$, respectively, and let us denote $D_i=\psi_i(D)$ the fixed points of $\psi_i$, where $i \in \set{1,2}$.
			
			Let us build $\psi: D_2 \rightarrow D_1$ equal to $\psi_1$ in $D_2$, i.e., for all $d \in D_2$ we set $\psi(d):=\psi_1(d)$. Since $\psi_1$ is a projection in $D$, $\psi$ is a projection in $D_2$ (the natural order is the same). Since $D_1$ is the fixed point of $\psi_1$ then $\psi_1(D_2) \subseteq D_1$. However, since $D_1 \subseteq D_2$ and $\psi_1(D_1) = D_1$ then $\psi_1(D_2)=D_1$, i.e., there is a projection $\psi$ such that $\psi_1=\psi\circ\psi_2$.
	\end{enumerate}
\end{proof}

\begin{example}
	Let us return to Example~\ref{ex:aggregated-length}. We change the threshold for the aggregated size. In Example~\ref{ex:aggregated-length} it was set to 2 ($\psi_{al=2}$), but we can change it to 5 ($\psi_{al=5}$) or 10 ($\psi_{al=10}$). The higher the threshold, the more possible descriptions are projected to themselves, i.e., belong to the fixed point of the projection. Thus, we have $\psi_{al=2} \leq \psi_{al=5} \leq \psi_{al=10}$.
\end{example}

Thanks to Proposition~\ref{prop:proj-and-concepts} it can be seen that, given a pattern structure $\PS$, if we have two projections $\psi_1 \leq \psi_2$, then the set of pattern extents of $\psi_1(\PS)$ is a subset of the set of pattern extents of $\psi_2(\PS)$, i.e., the smaller the projection, the smaller the number of concepts in the corresponding projected pattern structure.

Now it can be seen that projections actually form a semilattice with respect to the previously defined order.

\begin{proposition}
	Projections of a complete semilattice $(D,\sqcap)$ ordered by Definition~\ref{def:projections-order-fixed-point}~or~\ref{def:projections-order-superposition} form a semilattice $(\mathbb{F},\wedge)$, where the semilattice operation between  $\psi_1,\psi_2 \in \mathbb{F}$ is given by $\psi_1 \wedge \psi_2 = \psi_3$ iff $\psi_3(D)=\psi_1(D) \cap \psi_2(D)$.
\end{proposition}
\begin{proof}
	It follows from the definitions that if for any $\psi_1$ and $\psi_2$ the projection $\psi_3$ exists, then projections of $D$ form a semilattice. Let us describe the corresponding $\psi_3$.

	Let us denote $D_1=\psi_1(D)$ and $D_2=\psi_2(D)$ and $D_3=D_1 \cap D_2$. Let us suppose that there exist $x,y \in D_3$ such that $x \sqcup y \notin D_3$. But as $D_3 \subseteq D_1$ and $D_3 \subseteq D_2$, then, since $\psi_1$ is a projection of $D$ and $\psi_2$ is a projection of $D$, we have $x \sqcup y \in D_1$ and $x \sqcup y \in D_2$, i.e., $x \sqcup y \in D_1 \cap D_2 =D_3$. Thus, $(\forall x,y \in D_3) x \sqcup y \in D_3$. Then, according to Theorem~\ref{thm:suborder-of-a-semilattice} there is a projection $\psi_3$ such that $\psi_3(D)=D_3$.
\end{proof}

\subsection{Analogue of Theorem II for Revised Projections}
An important question is \textit{how a projection changes the representation context of a pattern structure?}
We limit the discussion of this question for the case when a set of description $D$ is a complete semilattice.
In~\cite{Ganter2001} the authors describe this change by means of Theorem 2. The formulation of this theorem was corrected in \cite{Kaiser2011}. Below we give the corrected version of the theorem.

\begin{theorem}[Theorem~2 from~\cite{Ganter2001}]
	For two pattern structures $(G,(D,\sqcap),\delta_1)$ and	$(G,(D,\sqcap),\delta_2)$ the following statements are equivalent:
	\begin{enumerate}
		\item
			$\delta_2=\psi\circ\delta_1$ for some $\psi$ on $(D,\sqcap)$.
		\item
			$(\forall g \in G)(\delta_2(g) \sqsubseteq \delta_1(g))$ and there is a representation context $(G,M,I)$ of $(G,(D,\sqcap),\delta_1)$ and some $N \subseteq M$ such that $(G,N,I \cap (G \times N))$ is a representation context of $(G,(D,\sqcap),\delta_2)$.
	\end{enumerate}
	\label{thm:ganter2001}
\end{theorem}

In Theorem~\ref{thm:ganter2001} one compares two pattern structures that differ in mapping functions. However, in the o-projected pattern structures we can modify the lattice structure itself. \textit{How can we adjust the formulation of Theorem~\ref{thm:ganter2001} in such a way that it can be applied to revised projections?} First, we should notice that in a pattern structure and in an o-projected pattern structure the set of objects is preserved. Second, the minimal representation context of a pattern structure can have less attributes than the minimal representation context of an o-projected pattern structure, as shown in Example~\ref{ex:increase-of-RC}.

\begin{example}\label{ex:increase-of-RC}
	Let $M=\{a,b,c\}$ and the description semilattice be $D=(2^M,\cap)$. Let $\psi:2^M \rightarrow 2^M$ be the following mapping: $\psi(\{a\})=\emptyset$ and for any $A \neq \{a\}$ we put $\psi(A)=A$. This projection is visualised in Figure~\ref{fig:proj-increasing-attrs-projection} by dashed arrows. Let us consider the following pattern structure $(\{g_1,g_2,g_3\},(2^M,\cap),\{g_1 \mapsto \{a,b\}, g_2 \mapsto \{a,c\}, g_3 \mapsto \{b,c\}\}$.

	The minimal representation context of this pattern structure contains 3 attributes $M=\{a,b,c\}$, while the minimal representation context of the o-projected pattern structure contains 4 attributes $M_\psi=\{b,c,ab,ac\}$. The corresponding contexts are shown in Figures~\ref{fig:proj-increasing-attrs-rc-ps}~and~\ref{fig:proj-increasing-attrs-rc-projected-ps}.
\end{example}

\begin{figure}[t]
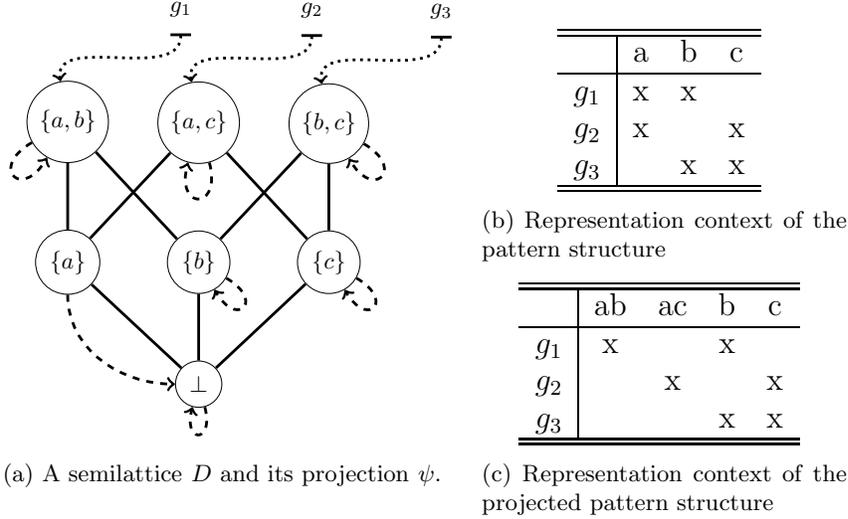

	\centering
	\begin{subfigure}[t]{0.55\columnwidth}
		\resizebox{\columnwidth}{!}{
			\putProjIncresingContext
		}
		\caption{A semilattice $D$ and its projection $\psi$.}
		\label{fig:proj-increasing-attrs-projection}
	\end{subfigure}
	\begin{minipage}[b]{0.4\columnwidth}
		\begin{subfigure}[b]{\columnwidth}
			\centering
			{\large
				\begin{tabular}{l|c@{~~~}c@{~~~}c}
				\hline\hline
				& a & b & c \\
				\hline
				$g_1$
				& x & x &   \\
				$g_2$
				& x &   & x \\
				$g_3$
				&   & x & x \\
				\hline\hline
			\end{tabular}
			}
			\caption{Representation context of the pattern structure}
			\label{fig:proj-increasing-attrs-rc-ps}
		\end{subfigure}

		\vspace{2mm}

		\begin{subfigure}[b]{\columnwidth}
			\centering
			{\large
			\begin{tabular}{l|c@{~~~}c@{~~~}c@{~~~}c}
				\hline\hline
				& ab & ac & b & c \\
				\hline
				$g_1$
				& x  &    & x &   \\
				$g_2$
				&    & x  &   & x \\
				$g_3$
				&    &    & x & x \\
				\hline\hline
			\end{tabular}
			}
			\caption{Representation context of the projected pattern structure}
			\label{fig:proj-increasing-attrs-rc-projected-ps}
		\end{subfigure}
		\vspace*{-14mm}
	\end{minipage}

	\vspace*{2mm}

	\caption{An example of a projection that can increase the number of attributes in the minimal represenation context.}
	\label{fig:proj-increasing-attrs}
\end{figure}

 We can see that to introduce the ``revised Theorem~2'' from~\cite{Ganter2001} we have to define a special relation between contexts.

\begin{definition}\label{def:smaller-context}
	Given two contexts $\K_1=(G,M_1,I_1)$ and $\K_2=(G,M_2,I_2)$, $\K_1$ is said to be simpler than $\K_2$, denoted by $\K_1 \leq_S \K_2$, if
	for any $m_{1,i} \in M_1$ there is a set $B_2\subseteq M_2$ such that $(\{m_{1,i}\})^1=(B_2)^2$. Here by $(\cdot)^1$ and $(\cdot)^2$ we denote the derivation operators in the contexts $\K_1$ and $\K_2$, respectively.
\end{definition}

\begin{example}
	The context in Figure~\ref{fig:proj-increasing-attrs-rc-projected-ps} is smaller w.r.t. Definition~\ref{def:smaller-context} than the context in Figure~\ref{fig:proj-increasing-attrs-rc-ps} because every column of the context in Figure~\ref{fig:proj-increasing-attrs-rc-projected-ps} is the intersection of a subset of columns of the context in Figure~\ref{fig:proj-increasing-attrs-rc-ps}.
\end{example}

This relation between contexts is a preorder. Indeed, it is reflexive, transitive, but not necessarily antisymmetric: given two contexts $\K_1$ and $\K_2$, if $\K_1$ and $\K_2$ have the same closure system of attributes, i.e., the same set of intents in the concept lattice, then according to the definition $\K_1 \leq_S \K_2$ and $\K_1 \geq_S \K_2$. However, we can consider only the context with the minimal number of attributes in the class of equivalence, i.e., the attribute-reduced context.  For simplicity in the rest of the paper we consider only attribute-reduced contexts.

This definition of the simplicity order on contexts can be related to context bonds~\cite{Ganter1999} in the following way. Three formal contexts $\K_i=(G_i,M_i,I_i)$ form a bond if $\K_1 \leq_S \K_2$ and $\K_2^T \leq_S \K_3^T$, where $\K^T=(M,G,I^T)$. Simplicity order can also be considered as a generalization of ``closed-relation-of" order  between contexts:


\begin{definition}[Definition~50 from~\cite{Ganter1999}]\label{def:closed-subrelation}
	A binary relation $J \subseteq I$ is called a \textbf{closed relation} of the context $(G,M,I)$ if every concept of the context $(G,M,J)$ is also a concept of $(G,M,I)$.
\end{definition}

From Definitions~\ref{def:smaller-context}~and~\ref{def:closed-subrelation} it can be seen that if $J$ is a closed relation of $(G,M,I)$, then $(G,M,J) \leq_S (G,M,I)$, but not always in the other direction.
The following theorem gives a relation between kernel operators of $D$ and the change in the representation context of o-projected pattern structures.

\begin{theorem} \label{thm:revised-theorem-II}
	Given a pattern structure $\PS=(G,(D,\sqcap),\delta)$ such that $(D,\sqcap)$ is a complete semilattice the following holds:
	\begin{enumerate}
		\item
			for any projection $\psi$ of $D$ we have $\rc(\psi(\PS)) \leq_S \rc(\PS)$.
		\item for any context $\K=(G,M,I)$ such that $\K \leq_S \rc(\PS)$, there is a projection $\psi$ of $D$ such that $\K$ is a representation context of $\psi(\PS)$.
	\end{enumerate}
\end{theorem}
\begin{proof}
	\begin{enumerate}
		\item
			The first statement follows from the fact that any extent of $\psi(\PS)$ is an extent of $\PS$ (Proposition~\ref{prop:proj-and-concepts}).
		\item
			Given a pattern structure $\PS$ and a context $\K$ such that $\K \leq_S \rc(\PS)$, let us define the set $D_M = \{d \in D \mid (\exists m \in M) (m')^{\diamond}=d\}$ (notice that for $\K$ and $\PS$ there is the same set $G$, thus, given $A \subseteq G$, both $A'$ and $A^\diamond$ are defined in $\K$ and $\PS$ correspondingly). Since $\K \leq_S \rc(\PS)$, $m'$ is an extent of $\PS$. Thus, we can see that there is a bijection between $D_M$ and $M$ given by $m' = d^{\diamond}$. We denote this bijection by $f(m)=d$, i.e. $f(m)=d \Leftrightarrow m'=d^\diamond$. Correspondingly, given a subset $N \subseteq M$, we denote by $f(N)=\{d \in D_M \mid f^{-1}(d) \in N \}$, i.e., $f(M)=D_M$.
			
			Let us define  $D_\psi = \{d \in D \mid (\exists X \subseteq D_M) \bigsqcup X = d\}$. According to Theorem~\ref{thm:suborder-of-a-semilattice} there is a projection $\psi$ such that $D_\psi=\psi(D)$.
			
			Let us consider the o-projected pattern structure $\psi(\PS)$. The set $D_M$ is $\sqcup$-dense for $\psi(D)$, i.e., the context $(G,D_M,I_{D_M})$, where $(g,d) \in I_{D_M} \Leftrightarrow \psi\circ\delta(g) \sqsupseteq d$, is a representation context of $\psi(\PS)$. There is the bijection between $D_M$ and $M$. Let us show that the relation $I$ is similar to the relation $I_M$, i.e., $(g,m) \in I \Leftrightarrow (g,f(m)) \in I_{D_M}$.

			It can be seen that for all $g \in G$ and all $d \in f(g')$, we get $\psi\circ\delta(g) \sqsupseteq d$, because for any $d \in f(g')$ we have $g \in d^\diamond$. Moreover, for any $\tilde{d} \in D \setminus f(g')$ we have $d \not\sqsupseteq\psi\circ\delta(g)$. Thus, the context $\K$ and the context $(G,D_M,I_{D_M})$ are similar, and hence for any context $\K \leq_S \rc(\PS)$ there is a projection such that $\K$ is a representation context of $\psi(\PS)$.
	\end{enumerate}
\end{proof}

\section{Conclusion}
In this paper we have introduced o-projections of pattern structures that are based on kernel operators $\psi: D \rightarrow D$. O-projections are a generalization of projections of pattern structures and allow one to change the semilattice of descriptions in o-projected pattern structures. Thus, the complexity of similarity (semilattice) operation can be reduced. Moreover, O-projections also correct a formal problem of projections.

We have shown that o-projections form a semilattice. This can be important when several independent o-projections are applied to a pattern structure. For example, if projections are discussed with several experts it may happen that several types of projections should be combined. In the case of several independent projections we know that there is the only one o-projection w.r.t. the semilattice of o-projections that is a combination of these projections.

Finally, we have shown that the representation context of an o-projected pattern structure can have more attributes than the representation context of the pattern structure itself. To describe this change in the representation context after o-projection we have introduced a new order on contexts, with the use of which we have described the way the representation context can change.

An important direction of the future work is to formalize \textit{transformations} of pattern structures, i.e., special homomorphisms between the semilattice of descriptions $D$ and a different semilattice $D_1$. In particular, it allows one to formalize the mappings of the form $\psi:D \rightarrow \mathbb{R}$, an instance of which are kernel functions used in Support Vector Machines (SVM).

\vspace{3mm}
{\noindent
\textbf{Acknowledgments:}
this research was supported by the Basic Research Program at the National Research University Higher School of Economics (Moscow, Russia) and by the BioIntelligence project (France). The second author was also supported by a grant from Russian Foundation for Basic Research, grant no. 13-0700504.
}

\bibliographystyle{plain}
\putBibliography

\end{document}